\documentclass[a4paper,UKenglish,cleveref, autoref, thm-restate]{lipics-v2021}

\usepackage[ruled]{algorithm2e}

\newcommand{\Geqt}{G_{\Delta}}

\pdfoutput=1 
\hideLIPIcs  

\title{Sublinear-Time Reconfiguration of Programmable Matter with Joint Movements} 

\titlerunning{Sublinear-Time Reconfiguration of Programmable Matter with Joint Movements} 

\author{Manish Kumar}{Indian Institute of Technology Ropar, India}{manish.kumar@njit.edu}{https://orcid.org/0000-0003-0620-3303}{}

\author{Othon Michail}{University of Liverpool, United Kingdom}{othon.michail@liverpool.ac.uk}{https://orcid.org/0000-0002-6234-3960}{}

\author{Andreas Padalkin}{Paderborn University, Germany}{andreas.padalkin@upb.de}{https://orcid.org/0000-0002-4601-9597}{}

\author{Christian Scheideler}{Paderborn University, Germany}{scheideler@upb.de}{https://orcid.org/0000-0002-5278-528X}{}

\authorrunning{M. Kumar, O. Michail, A. Padalkin, and C. Scheideler}

\ccsdesc[500]{Computing methodologies~Cooperation and coordination}
\ccsdesc[300]{Theory of computation~Computational geometry}

\keywords{amoebot model, programmable matter, modular robot system, reconfiguration} 

\funding{\textit{Andreas Padalkin} and \textit{Christian Scheideler}: These authors were supported by the DFG Project SCHE 1592/10-1.}
\acknowledgements{We thank Christian Bauer and Leo Decking for implementing the presented algorithms.}

\nolinenumbers 

\begin{document}

\maketitle
\begin{abstract}
We study centralized reconfiguration problems for geometric amoebot structures. A set of $n$ amoebots occupy nodes on the triangular grid and can reconfigure via expansion and contraction operations. We focus on the joint movement extension, where amoebots may expand and contract in parallel, enabling coordinated motion of larger substructures. Prior work introduced this extension and analyzed reconfiguration under additional assumptions such as metamodules.

In contrast, we investigate the intrinsic dynamics of reconfiguration without such assumptions by restricting attention to centralized algorithms, leaving distributed solutions for future work. We study the \emph{reconfiguration problem} between two classes of amoebot structures $A$ and $B$: For every structure $S\in A$, the goal is to compute a schedule that reconfigures $S$ into some structure $S'\in B$. Our focus is on sublinear-time algorithms.

We affirmatively answer the open problem by Padalkin \emph{et al.} (Auton. Robots, 2025) whether a within-the-model sublinear-time universal reconfiguration algorithm is possible, by proving that any structure can be reconfigured into a \emph{canonical} line-segment structure in $O(\sqrt{n}\log n)$ rounds. Additionally, we give a constant-time algorithm for reconfiguring any spiral structure into a line segment. These results are enabled by new constant-time primitives that facilitate efficient parallel movement.

Our findings demonstrate that the joint movement model supports sublinear reconfiguration without auxiliary assumptions. A central open question is whether universal reconfiguration within this model can be achieved in polylogarithmic or even constant time.
\end{abstract}


\section{Introduction}

Programmable matter is made of many small, identical robotic modules that can change the properties of matter, such as its shape, color, or density, in a programmable way~\cite{DBLP:journals/ijhsc/ToffoliM93}. 
In this paper, we consider the reconfiguration dynamics of the geometric amoebot model~\cite{DBLP:series/lncs/DaymudeHRS19,DBLP:journals/dc/DaymudeRS23,DBLP:conf/spaa/DerakhshandehDGRSS14} from a centralized perspective. 
In the distributed version of this model, a set of $n$ robots (called \emph{amoebots}) is placed on a triangular grid. 
Each amoebot can communicate with its neighboring amoebots and move by expanding into an empty adjacent node and then contracting into it. 
Since movement and communication happen locally from node to node, many problems in this model have a natural lower bound of $\Omega(D)$ rounds, where $D$ is the diameter of the structure.

To overcome this limitation, previous work introduced the \emph{joint movement extension}~\cite{DBLP:journals/arobots/PadalkinKS25}.
This extension allows multiple amoebots to move together in a coordinated way.
This extension allows multiple amoebots to move together in a coordinated way.
For example, an amoebot can push or pull other amoebots while expanding or contracting, which enables large parts of the structure to move in parallel.
This raises an important question: Can we reconfigure an amoebot structure much faster using joint movements? By ``faster'' we mean running in sublinear parallel time (measured in discrete rounds), or, if possible, even in polylogarithmic or constant time. In this paper, we are interested in whether this is in principle feasible, thus we naturally focus on the centralized reconfiguration dynamics of the model.

Actuation operations that can globally affect the structure are sometimes called \emph{linear-strength operations}, representing the property that they have sufficient strength to displace up to a linear-sized set.
Although such operations may be a strict assumption to make in some systems, they have been incorporated in a variety of theoretical models \cite{almalki2024geometric,almethen2020pushing,DBLP:journals/comgeo/AloupisCDDFLORAW09,DBLP:conf/isaac/AloupisCDLAW08,kostitsyna2024turning,woods2013active} and are reasonable in several real-world contexts like micro- and nano-robotic systems in low-viscosity environments, DNA nanotechnology, and programmable matter where a single local operation can directly cause reconfiguration of the entire system.
They also form a simple theoretical abstraction for dynamics that, in practice, could be implemented by joint movemements, with multiple robots jointly bearing the load of a linear-strength operation. We use the terms \emph{move} and \emph{movement} to refer to actuation operations throughout the paper.

Aloupis \emph{et al.} \cite{DBLP:journals/comgeo/AloupisCDDFLORAW09,DBLP:conf/isaac/AloupisCDLAW08} studied a model of a modular robot system known as \emph{crystalline robots} \cite{rus2001crystalline}. 
The 2D version of the crystalline model represents modules as squares on a 2D grid, forming a connected shape of modules attached to adjacent modules. Each individual module can expand and contract, by extending one of its faces one unit out and retracting it back at some later point. Due to modules being attached to each other, up to linear-size components can move due to a module's expansion or contraction. 
In \cite{DBLP:conf/isaac/AloupisCDLAW08}, they gave a universal centralized reconfiguration algorithm for the crystalline model that, for any pair of connected shapes $S_I, S_F$ of the same number of modules $n$, can transform $S_I$ into $S_F$ in $O(\log n)$ rounds. 
Note that in the crystalline model reconfiguration of some classes of shapes is impossible without metamodules (or other assumptions), and these algorithms essentially rely on the use of metamodules.

Recent work has showed that joint movements in the geometric amoebot model can also be very powerful when additional assumptions such as metamodules are used~\cite{DBLP:journals/arobots/PadalkinKS25}.
However, it is still unclear how powerful the model is on its own, without relying on extra assumptions.
In this paper, we study the centralized reconfiguration problem in the geometric amoebot model with joint movements, without relying on such additional assumptions, thus focusing on the limits of the model dynamics themselves.

We consider centralized algorithms, where a global scheduler knows the full structure and decides the movements of all amoebots in each round. 
This allows us to focus on the main question: How fast can reconfiguration be in principle? 
Given two classes of amoebot structures, the goal is to design sequences of movements that reconfigure every structure from one class into some structure of the other class. An algorithm in this context is a formal description of these sequences of movements for all structures in the source class.

Our main result is the first within-the-model universal sublinear-time reconfiguration algorithm. 
We show that any connected amoebot structure can be reconfigured into a line segment in $O(\sqrt{n}\log n)$ rounds. 
Since any two connected structures with the same number of amoebots can both be reconfigured into a line segment, and every reconfiguration is reversible, this immediately implies that we can reconfigure any structure into any other structure within the same asymptotic bound. 
Therefore, we obtain universal reconfiguration in $O(\sqrt{n}\log n)$ rounds. 
This answers an open question on whether large-scale reconfiguration can be achieved in sublinear time without additional assumptions \cite{DBLP:journals/arobots/PadalkinKS25}.
We also present an algorithm that reconfigures spiral structures into a line segment in constant time.
This result shows that non-trivial subclasses of paths can be reconfigured extremely fast within the model.
We hope that this can serve as a foundational step toward polylogarithmic or even constant-time universal reconfiguration, which is left as a central open question.

Overall, our results show that joint movements can greatly increase the \emph{reconfiguration efficiency} of programmable matter structures, i.e., how quickly a structure can be reconfigured into another. While the original amoebot model is limited by diameter-based lower bounds, as we show, coordinated parallel motion enables significantly faster global reconfiguration.

\textbf{Organization of the paper.}
In \Cref{sec:Model}, we describe the geometric amoebot model with joint movements and formally define the reconfiguration problem and the classes of structures considered in this paper. 
\Cref{sec:preliminaries} introduces a set of basic movement primitives that will be used as building blocks in our algorithms.
In \Cref{sec:monotone}, we present algorithms for reconfiguring special classes of structures, including monotone, histogram, and star-convex structures, into simpler canonical forms such as a line segment.
\Cref{sec:universal} describes our main result: a universal reconfiguration algorithm that reconfigures any connected structure into a line segment in sublinear time.
\Cref{sec:spirals} presents algorithms for reconfiguring spiral structures into a line segment, including a constant-time solution.
Finally, \Cref{sec:conclusion} discusses open problems.


\section{Model and Problem Statement}
\label{sec:Model}

In this section, we formally present the model and the problem statement.
In \Cref{sec:model:extension}, we present the geometric amoebot model with joint movements.
In \Cref{sec:classes}, we define the reconfiguration problem and the classes of amoebot structures considered in this paper.


\subsection{Geometric Amoebot Model with Joint Movements}
\label{sec:model:extension}



In this section, we present the geometric amoebot model \cite{DBLP:journals/dc/DaymudeRS23,DBLP:conf/spaa/DerakhshandehDGRSS14} with the joint movement extension \cite{DBLP:journals/arobots/PadalkinKS25}.
%
The amoebot structure consists of $n$ amoebots placed on the infinite regular triangular grid graph $\Geqt = (V_\Delta, E_\Delta)$ (see the left side of \Cref{fig:model}).
Each amoebot either occupies a single node or two adjacent nodes and the edge between them.
We call an amoebot \emph{contracted} if the amoebot occupies a single node and \emph{expanded} otherwise.
Every node of $\Geqt$ is occupied by at most one amoebot.
Adjacent amoebots are connected by \emph{bonds}.

\begin{figure}[tb]
    \centering
    \includegraphics{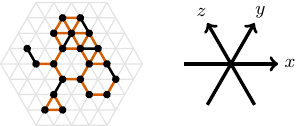}
    \caption{
        Amoebot model and axes.
        The black nodes and edges indicate the nodes and edges occupied by amoebots.
        The red edges indicate bonds between the amoebots.
        The gray edges indicate the triangular grid $\Geqt$.
        We omit the grid in all other figures.
    }
    \label{fig:model}
\end{figure}


Let $C$ denote the set of nodes occupied by contracted amoebots, $X$ the set of nodes occupied by expanded amoebots, $A$ denote the set of all edges occupied by expanded amoebots, and $B$ the set of all bonds.
We define the \emph{connectivity graph} of an amoebot structure as the graph $S = (C \cup X, A \cup B)$.
The edges of $A$ are pairwise disjoint.
We assume that initially, $S$ is connected.
%
%
To each edge in $A \cup B$, we assign a geometric orientation.
We only allow orientations parallel to the axes of the triangular grid $\Geqt$ (see the right side of \Cref{fig:model}).
The \emph{amoebot structure} is defined by its connectivity graph and the assignment of orientations.
For the sake of simplicity, we will denote each amoebot structure by its connectivity graph.


We assume the \emph{fully synchronous} activation model, i.e., time is divided into synchronous rounds, and every amoebot is active in each round.
In each round, each amoebot can perform a single movement which is performed in two steps.
Let $S_i = (C_i \cup X_i, A_i \cup B_i)$ be the amoebot structure at the beginning of the round.


In the \textbf{first step},
the amoebots remove bonds from $S_i$ as follows.
Each amoebot can decide to release an arbitrary subset of its currently incident bonds in $S_i$.
A bond is removed if and only if one of the incident amoebots releases the bond.
Note that we only remove bonds (i.e., edges in $B_i$) but no occupied edges (i.e., edges in $A_i$).
Let $R_i \subseteq B_i$ be the set of the remaining bonds and $S'_i = (C_i \cup X_i, A_i \cup R_i)$ the resulting amoebot structure.
We require that $S'_i$ is connected since otherwise, disconnected parts might float apart.
We say that a \emph{connectivity conflict} occurs if and only if $S'_i$ is not connected.
Whenever a connectivity conflict occurs, the amoebot structure transitions into an undefined state such that we become unable to make any statements about the structure.


In the \textbf{second step},
each amoebot may perform one of the following movements.
A contracted amoebot $c$ occupying $c_0 \in C_i$ may \emph{expand} on one of the axes of the grid (see \Cref{fig:model:expansion}).
For that, we replace $c_0$ by two nodes $c_1,c_2$ and an edge $a = \{ c_1, c_2 \}$, and assign an orientation to $a$.
For each to $c$ incident bond $b \in R_i$, we replace $c_0$ with one of the new nodes.
Note that the incident bonds do not change their orientations.
As a result, all connected amoebots move with the expanding amoebot.
An expanded amoebot may \emph{contract} analogously by reversing the contraction (see \Cref{fig:model:contaction}).

\begin{figure}[tb]
    \begin{minipage}[t]{.29\linewidth}
        \centering
        \includegraphics{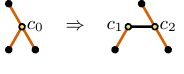}
        \subcaption{Expansion.}
        \label{fig:model:expansion}
    \end{minipage}
    \hfill
    \begin{minipage}[t]{.29\linewidth}
        \centering
        \includegraphics{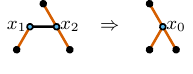}
        \subcaption{Contraction.}
        \label{fig:model:contaction}
    \end{minipage}
    \hfill
    \begin{minipage}[t]{.38\linewidth}
        \centering
        \includegraphics{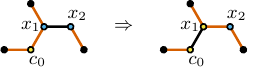}
        \subcaption{Handover.}
        \label{fig:model:handover}
    \end{minipage}
    \caption{
        Movements.
        The yellow amoebot in (a) and (c) indicates amoebot $c$.
        The blue amoebot in (b) and (c) indicates amoebot $x$.
    }
\end{figure}

Furthermore,
pairs of amoebots may perform isolated \emph{handovers} as follows (see \Cref{fig:model:handover}).
Consider a contracted amoebot $c$ occupying $c_0 \in C_i$ and an expanded amoebot $x$ occupying $x_1,x_2 \in X_i$ and $a = \{ x_1, x_2 \} \in A_i$ that are connected by a bond $b = \{ c_0, x_1 \} \in R_i$.
Intuitively, we want to switch the association of the expanded amoebot and the bond.
More precisely,
(i) amoebot $c$ becomes expanded and occupies nodes $c_0$, $x_1$, and edge $b$,
(ii) amoebot $x$ becomes contracted and occupies node $x_2$, and
(iii) edge $a$ becomes a bond.

At the end of the movements, let $C_{i+1}$ denote the set of nodes occupied by contracted amoebots, and $X_{i+1}$ the set of nodes occupied by expanded amoebots, $A_{i+1}$ the set of all edges occupied by expanded amoebots, and $R'_i$ the set of all bonds.
Let $S''_i = (C_{i+1} \cup X_{i+1}, A_{i+1} \cup R'_i)$ be the resulting connectivity graph.
We require that $S''_i$ is a subgraph of $\Geqt$ in compliance with the orientations of all edges (i.e., $A_{i+1} \cup R'_i$).
We say that a \emph{collision} occurs if and only if $S''_i$ is not a subgraph of $\Geqt$, i.e., either the amoebots cannot be mapped to the triangular grid or multiple amoebots are mapped to the same node.
Whenever a collision occurs, the amoebot structure transitions into an undefined state such that we become unable to make any statements about the structure.
We, as algorithm designers, are responsible to make sure collisions do not happen.
In the next round, we continue with the \emph{adjacency closure} $S_{i+1} = (C_{i+1} \cup X_{i+1}, A_{i+1} \cup B_{i+1})$ of $S''_i$, i.e., we add bonds between amoebots until all adjacent amoebots are connected.




In this paper, we assume that we have a centralized scheduler.
The scheduler knows the current state of the amoebot structure at all times.
At the beginning of each synchronous round, it decides for each amoebot which bonds to release, and which movement to perform.
Due to our focus being on understanding the feasibility of sublinear-time dynamics within the model (i.e., no additional assumptions in contrast to previous work which used metamodules), we leave the design of distributed solutions for future work.



\subsection{Problem Statement and Classes of Amoebot Structures}
\label{sec:classes}


In this paper, we consider the \emph{reconfiguration problem} between two classes $A$ and $B$ of amoebot structures of the same size.
The goal of the problem is to develop an algorithm that for every amoebot structure $S \in A$ computes a schedule that reconfigures $S$ into any amoebot structure $S' \in B$.
Note that an algorithm may map the amoebot structures of $A$ to only a subset of $B$ and not the whole class.
Hence, in general, it is not possible to use the same algorithm with reversed schedules to solve the reverse problem.
In the remainder of this paper, we will consider the following classes (see \Cref{fig:classes}, which also shows our results).
We assume that all structures in the classes only consist of contracted amoebots.
Note that each class is closed under translation, rotation, and reflection.

\begin{figure}[tb]
    \centering
    \includegraphics[scale=.55]{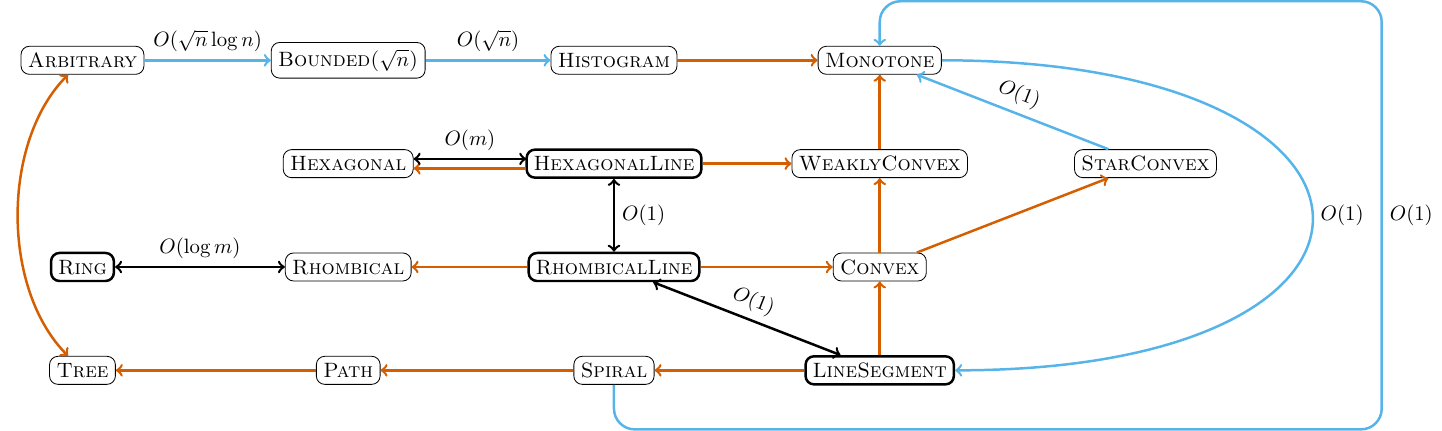}
    \caption{
        Results.
        Classes with a thicker frame indicate canonical classes.
        Red arrows point to superclasses and with that indicate trivial solutions.
        Note that we do not depict all superclasses.
        Black arrows indicate results by Padalkin \emph{et al.} \cite{DBLP:journals/arobots/PadalkinKS25}.
        Note that their superconstant algorithms simulate algorithms from \cite{DBLP:conf/isaac/AloupisCDLAW08,DBLP:journals/arobots/HurtadoMRA15}.
        They define rhombical and hexagonal metamodules.
        {\sc Rhombical} ({\sc Hexagonal}) contains all amoebot structure that consist of rhombical (hexagonal) metamodules.
        {\sc RhombicalLine} ({\sc HexagonalLine}) contains all amoebot structure forming a line of rhombical (hexagonal) metamodules.
        {\sc Ring} contains canonical amoebot structures forming a ring (see \cite{DBLP:conf/isaac/AloupisCDLAW08} for details).
        $m$ denotes the number of metamodules.
        Blue edges indicate our results.
        Note that our algorithm that reconfigures bounded structures into histograms simulates an algorithm from \cite{DBLP:journals/comgeo/AloupisCDDFLORAW09}.
    }
    \label{fig:classes}
\end{figure}

\smallskip


\noindent{\sc Arbitrary}:
All structures.
\smallskip


\noindent{\sc Tree}:
Structures forming a tree.
Note that this class is equal to {\sc Arbitrary} since each structure has a spanning tree.
\smallskip


\noindent{\sc Path}:
Structures forming a path.
\smallskip


\noindent{\sc Spiral}:
Structures forming a simple spiral with $120^\circ$ corners.
Note that a single line segment is a one-segmented spiral.
\smallskip


\noindent{\sc LineSegment}:
Structures that form a line segment.
\smallskip


\noindent{\sc Monotone}:
We call the intersection of the amoebot structure with a line parallel to the $x$-axis an \emph{$x$-section}.
Note that an $x$-section is not necessarily connected.
We define \emph{$y$- and $z$-sections} analogously.
In particular, we also call $x$-sections \emph{rows}, and $y$-sections \emph{columns}.
We call an amoebot structure \emph{$x$-monotone} if and only each $x$-section is connected.
We define $y$- and $z$-monotone analogously.
We call an amoebot structure \emph{monotone} if and only if it is $x$-monotone, $y$-monotone, or $z$-monotone.
\smallskip


\noindent{\sc Histogram}:
We call a line segment of amoebots parallel to the $x$-axis an \emph{$x$-segment}.
We define \emph{$y$- and $z$-segments} analogously.
We say that an amoebot structure forms a \emph{histogram} if and only if it consists of an $x$-segment with $y$-segments attached to one side of the $x$-segment, or can be obtained from such a structure by reflections and rotations.
Note that every histogram is monotone.
\smallskip


\noindent{\sc Bounded($k$)}:
Structures with at most $k$ non-empty $x$-, $y$-, or $z$-sections.
\smallskip


\noindent{\sc Convex}:
We call an amoebot structure \emph{convex} if and only if for each pair of amoebots $u,v$, all shortest paths between $u$ and $v$ are part of the amoebot structure.
\smallskip


\noindent{\sc WeaklyConvex}:
We call an amoebot structure \emph{weakly convex} if and only if for each pair of amoebots $u,v$, at least one shortest path between $u$ and $v$ is part of the amoebot structure.
\smallskip


\noindent{\sc StarConvex}:
We call an amoebot structure \emph{star-convex} if and only if there is an amoebot $u$ such that for each amoebot $v$, all shortest paths between $u$ and $v$ are part of the amoebot structure \cite{DBLP:conf/wdag/ArtmannPS25}.
\smallskip


Another interesting variant of the \emph{reconfiguration problem} is the one within a class $A$.
The goal of this variant is to develop an algorithm that for any two amoebot structures $S, S' \in A$ (with the same number of amoebots) computes a schedule that reconfigures $S$ into $S'$.
We obtain a universal reconfiguration algorithm if we solve the problem for the class {\sc Arbitrary}.
A common way to design such an algorithm is to use a canonical structure $C$ as an intermediate structure during the reconfiguration (e.g., \cite{DBLP:journals/comgeo/AloupisCDDFLORAW09,DBLP:conf/isaac/AloupisCDLAW08,DBLP:journals/arobots/HurtadoMRA15}).
Given schedules from $S$ and $S'$ to $C$, we can obtain a schedule from $S$ to $S'$ by first executing the schedule from $S$ to $C$ and then the reverse schedule from $S'$ to $C$.
In this paper, we utilize a line segment as the canonical structure.
Our focus therefore is to solve the reconfiguration problem between various classes and the class {\sc LineSegment}.
Our main results are a universal reconfiguration algorithm that requires $O(\sqrt n \log n)$ rounds and an algorithm that reconfigures a spiral into a line in $O(1)$ rounds.
Further results can be found in \Cref{fig:classes}.



\section{Preliminaries}
\label{sec:preliminaries}

In this section, we present basic movement primitives: the tunneling, shearing, parallelogram, triangle, and trapezoid primitive.
These are core components in our reconfiguration algorithms.

In the tunneling primitive, originally defined in \cite{DBLP:conf/dna/DerakhshandehGS15}, we move a chain of amoebots along a path by constantly performing handovers.%
\footnote{In previous work, the primitive is defined for spanning forests and is therefore called the spanning forest primitive. Since we only need the simpler case of a path, we refer to \cite{DBLP:series/lncs/DaymudeHRS19,DBLP:conf/dna/DerakhshandehGS15} for the general case.}
\Cref{alg:tunneling} shows the algorithm performed by each amoebot.
We will use this primitive to ``tunnel'' an amoebot from one end to the other end without changing the structure of the remaining chain.
In general, it may take up to $O(n)$ rounds to tunnel the amoebot through the chain.
We call a chain that alternates between contracted and expanded amoebots \emph{alternating}.
On alternating chains, the primitive only requires $3$ rounds (see \Cref{fig:tunneling}).
Note that the primitive actually moves all amoebots along the path without changing the order.
However, since the amoebots are identical, the outcome is equivalent to a real tunneling.
By construction, we obtain the following lemma.


\begin{algorithm}[tb]
\caption{Tunneling Primitive \cite{DBLP:series/lncs/DaymudeHRS19}}
\label{alg:tunneling}
\If{contracted}{
    \If{last amoebot of chain}{
        contract
    }\Else{
        perform handover with the predecessor
    }
}\Else{
    \If{first amoebot of chain}{
        expand along the path
    }\Else{
        perform handover with the successor
    }
}
\end{algorithm}


\begin{figure}[tb]
    \centering
    \includegraphics{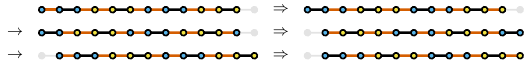}
    \caption{
        Tunneling primitive.
        In each round, adjacent amoebots of the same color perform a handover.
        The gray node is unoccupied.
        In all figures, a simple arrow between two subfigures indicates the execution of the first step of a round (i.e., the removal of bonds) after establishing the adjacency closure, and a double arrow the execution of the second step of a round (i.e., the movements) unless stated otherwise.
        Note that only the fully completed rounds (i.e., both steps were executed) count to the total number of rounds.
    }
    \label{fig:tunneling}
\end{figure}


\begin{lemma}
\label{lem:tunneling}
    On alternating chains, the tunneling primitive requires $3$ rounds to tunnel an amoebot through the chain.
\end{lemma}


We perform all remaining primitives on substructures of the amoebot structure.
Other parts of the amoebot structure may be connected to that substructure at specific amoebots which we call \emph{connection points}.
Each primitive has two of these connection points.
As the parts are not necessarily connected without the substructure, we must keep it connected at all times.
Hence, we must keep the substructure connected at all times.
With the exception of the shearing primitive, we maintain the relative position between those connection points.


In the shearing primitive, we shear a line segment of amoebots to another axis.%
\footnote{%
This primitive is similar to the realignment primitive of \cite{DBLP:journals/arobots/PadalkinKS25}.
In fact, both primitives use the same intermediate structure.
This allows us to switch between all structures of both primitives.
}
We emphasize that this is not a rotation of the line segment since the attached parts of the amoebot structure are not rotated with the line segment.
The endpoints of the line serve as the connection points.

We realize the primitive as follows (see \Cref{fig:shearing}).
The primitive consists of two rounds.
In the first round, we shear the line halfway by expanding every second amoebot.
In the second round, we mirror the first round to shear the line segment to the new axis.
Note that from the perspective of a connection point, the shearing primitive is performed within a triangle defined by the initial and final location of the line segment (see the yellow area in \Cref{fig:shearing}).
By construction, we obtain the following lemma.

\begin{figure}[tb]
    \centering
    \includegraphics{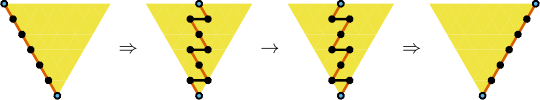}
    \caption{
        Shearing primitive.
        The blue amoebots indicate the connection points.
        The yellow area indicates the triangle within the primitive is performed from the perspective of the bottom connection points.
    }
    \label{fig:shearing}
\end{figure}

\begin{lemma}
\label{lem:shearing}
    Our implementation of the shearing primitive requires $2$ rounds.
\end{lemma}

\begin{remark}
    The shearing primitive can be generalized to parallelograms of contracted amoebots.
\end{remark}



Consider a parallelogram where only the two sides incident to an obtuse corner are occupied by amoebots.
Let $\ell_1$ and $\ell_2$ denote the side lengths of the parallelogram such that w.l.o.g., $\ell_1 \leq \ell_2$.
The amoebots at the acute corners serve as the connection points.
The goal of the parallelogram primitive is to move the amoebots to the other two sides without leaving the parallelogram and without changing the relative positions of the connection points at any time.
Note that these sides require the same number of amoebots as the initially occupied sides.

We realize the primitive as follows (see \Cref{fig:parallelogram}).
On both sides of the parallelogram, we apply the shearing primitive on the the first $\ell_1 + 1$ amoebots starting from the connection point, respectively.
Note that in case the parallelogram is a rhombus, the $(\ell + 1)$-st amoebot of both sides is identical and participates in both shearing primitives.
This is possible since the necessary movements of both shearing primitives match.
By construction, we obtain the following lemma.

\begin{figure}[tb]
    \centering
    \includegraphics{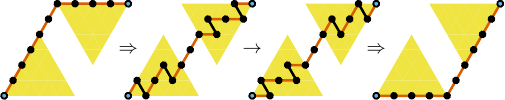}
    \caption{
        Parallelogram primitive.
        The blue amoebots indicate the connection points.
    }
    \label{fig:parallelogram}
\end{figure}

\begin{lemma}
\label{lem:parallelogram}
    Our implementation of the parallelogram primitive requires $2$ rounds.
\end{lemma}


Consider an equilateral triangle of side length $\ell$ where only two sides (w.l.o.g., the legs) are occupied by amoebots.
The amoebots at the corners incident to the unoccupied side (i.e., the base) serve as the connection points.
The goal of the triangle primitive is to occupy the base and one of the legs without leaving the triangle and without changing the relative positions of the connection points at any time.
Note that any two sides of the triangle require the same number of amoebots since the triangle is equilateral.


We cannot simply shear one leg to the base without disconnecting the amoebots within the triangle (e.g., by shearing one leg directly to the base) or changing the relative positions of the connection points (e.g., by first shearing one leg to the other leg so that both legs become adjacent and then shearing it to the base).
Instead, we realize the primitive as follows.
W.l.o.g., we assume that the legs are aligned to the $y$- and $z$-axis.
Our implementation consists of $4$ phases.
We first outline the phases.
In the first phase, we iteratively reduce the size of the triangle (i.e., the length of its sides) by $1$ until $\ell = 1$ or there is an even $k$ such that $\ell = 4 k + 1$.
In the second phase, we occupy the base.
However, for that, we move out of both legs.
In the third phase, we occupy all nodes adjacent to the base within the triangle.
In the fourth phase, we move those amoebots to one of the legs.


Consider the \textbf{first phase}, i.e., the iterative reduction of the triangle (see \Cref{fig:triangle:1}).
We pair the amoebots on one leg starting from the top except for the last $2$ amoebots if $\ell$ is odd and the last amoebot if $\ell$ is even.
Then, we use the tunneling primitive on each pair to fill the nodes adjacent to the leg.
We place the two excess amoebots at the bottom of the original leg.
This will become important in the fourth phase. We ignore them for now.

We stop once $\ell = 1$ or there is an even $k$ such that $\ell = 4 k + 1$.
If the former holds, we directly go to the fourth phase.
Otherwise, we continue with the second phase.

\begin{figure}[tb]
    \begin{minipage}[t]{\linewidth}
        \centering
        \includegraphics{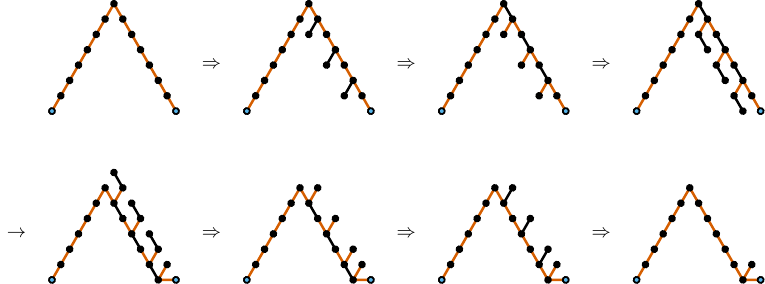}
        \subcaption{Case $\ell$ is odd.}
    \end{minipage}
    
    \begin{minipage}[t]{\linewidth}
        \centering
        \includegraphics{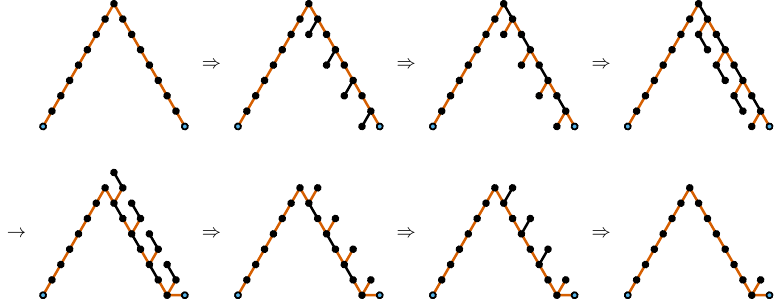}
        \subcaption{Case $\ell$ is even.}
    \end{minipage}
    \caption{
        First phase of the triangle primitive.
        The blue amoebots indicate the connection points.
    }
    \label{fig:triangle:1}
\end{figure}


Consider the \textbf{second phase}, i.e., the occupation of the base (see \Cref{fig:triangle:2}).
In each leg, we apply the shearing primitive on two subsegments.
Let $(v_1, \dots, v_\ell = v_{4k+2})$ denote the amoebots of a leg from bottom to top.
In the left leg, we shear $(v_1, \dots, v_{2k+1})$ to align it to the $x$-axis, and $(v_{2k+1}, \dots, v_{3k+1})$ to align it to the $z$-axis.
In the right leg, we shear $(v_1, \dots, v_{2k+1})$ to align it to the $x$-axis, and $(v_{2k+1}, \dots, v_{3k+1})$ to align it to the $y$-axis.
The lengths of the subsegments in both legs are chosen in a way that they even out their horizontal movement.
This implies that the remaining amoebots are moved straight downwards and that the connection points do not change their relative positions.
Since $\ell = 4 k + 1$, all nodes of the base are occupied after the phase.
Furthermore, the remaining amoebots form a rhombus with a missing corner on top of the base.
We continue with the next phase.

\begin{figure}[tb]
    \centering
    \includegraphics{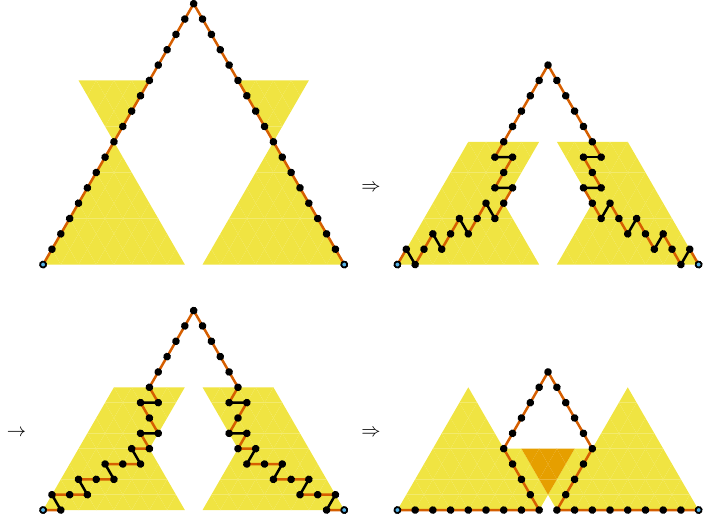}
    \caption{
        Second phase of the triangle primitive on a triangle of side length $\ell = 17$.
        Note that $k = 4$ holds.
        The blue amoebots indicate the connection points.
    }
    \label{fig:triangle:2}
\end{figure}


Consider the \textbf{third phase}, i.e., the occupation of the nodes adjacent to the base (see \Cref{fig:triangle:3}).
Since the base is occupied, we can disconnect the rhombus at the top to obtain two ``arms''.
We use the shearing primitive on both arms to move them to the base.
We start with the shorter arm.
Then, we shift the longer arm to the center by an expansion and contraction.
this allows us to shear it without leaving the triangle.
Now, the base and all adjacent nodes are occupied.
Note that this also holds for the excess amoebots from the first phase.
We continue with the next phase.

\begin{figure}[tbp]
    \centering
    \includegraphics{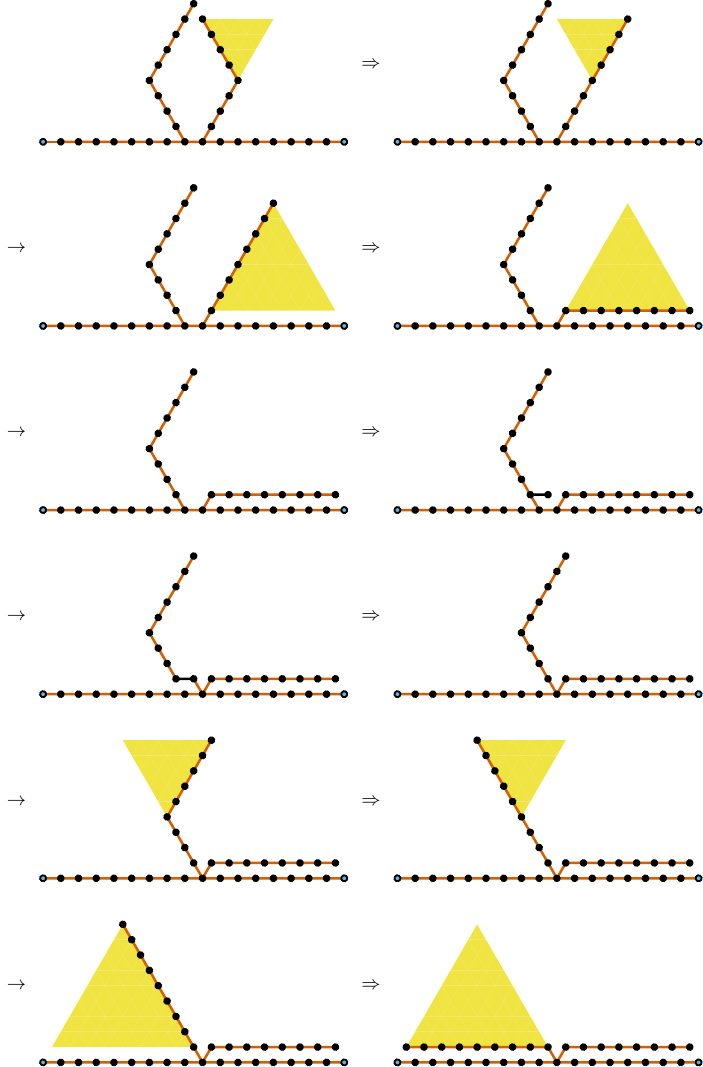}
    \caption{
        Third phase of the triangle primitive.
        The blue amoebots indicate the connection points.
        The first two and last two double arrows indicate $2$ rounds.
    }
    \label{fig:triangle:3}
\end{figure}


Consider the \textbf{fourth phase}, i.e., the occupation of one of the legs (see \Cref{fig:triangle:4}).
Here, we simply apply another shearing primitive on the amoebots adjacent to the base to move them to one of the legs.
By construction, we obtain the following lemma.

\begin{figure}[tb]
    \centering
    \includegraphics{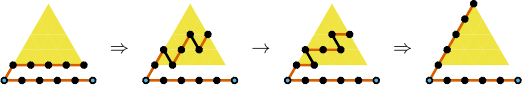}
    \caption{
        Fourth phase of the triangle primitive.
        The blue amoebots indicate the connection points.
    }
    \label{fig:triangle:4}
\end{figure}

\begin{lemma}
\label{lem:triangle}
    Our implementation of the triangle primitive requires $O(1)$ rounds.
\end{lemma}


Finally, consider a trapezoid where the shorter base and the legs are occupied by amoebots.
The amoebots at the acute corners serve as the connection points.
Let a node on the longer base be given.
We will call it the \emph{starting point}.
The goal of the trapezoid primitive is to occupy the longer base and a path between the bases starting from the starting point, which may consist of up to $2$ segments, without leaving the trapezoid and without changing the relative positions of the connection points at any time.
Note that the the shorter base and both legs require the same number of amoebots as the longer base and one leg.
Furthermore, the path between the legs requires the same number of amoebots as one leg.
Hence, we have the exact number of amoebots for this reconfiguration.


We realize the primitive as follows (see \Cref{fig:trapezoid}).
Our implementation consists of 4 phases.
In the \textbf{first phase}, we split the trapezoid into a triangle and a parallelogram such that the starting point is in the triangle.
Then, we apply the parallelogram primitive on the parallelogram.
This partially occupies the longer base.
It remains to deal with the resulting triangle.


In the \textbf{second phase}, we distinguish between two cases.
If the starting point is the corner of the triangle, we apply the triangle primitive on the triangle such that on completion, the occupied leg starts at the starting point.
This occupies the remaining part of the longer base.
Furthermore, the occupied leg forms a path between both bases starting from the starting point (i.e., the path consists of a single segment).
At this point, we terminate the primitive.

Otherwise, i.e., if the starting point is not the corner of the triangle, we split the triangle into two triangles and a parallelogram such that the starting point is a corner of the lower triangle and parallelogram.
Then, we apply the triangle primitive on the upper triangle.
We obtain a trapezoid (consisting of the lower triangle and the parallelogram) with an attached leg at one of its corners.


In the \textbf{third and fourth phase}, we repeat the first and second phase while keeping the attached leg attached.
However, this time, we can terminate since the starting point is one of the corners of the triangle.
Note that the resulting legs from the second and fourth phase form a path between the bases starting from the starting point (i.e., the path consists of $2$ segments).
By construction, we obtain the following lemma.

\begin{figure}[tb]
    \centering
    \includegraphics{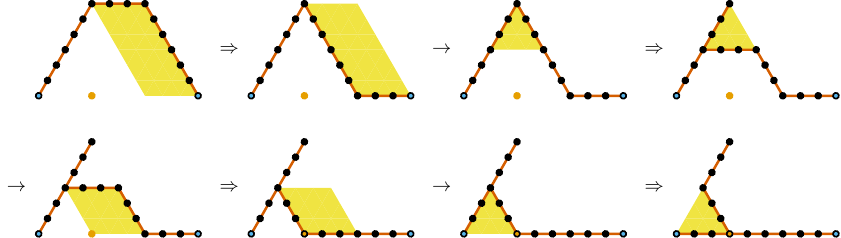}
    \caption{
        Trapezoid primitive.
        The orange node/amoebot indicates the starting point.
        The blue amoebots indicate the connection points.
        Each double arrow indicates a phase consisting of $2$ rounds.
    }
    \label{fig:trapezoid}
\end{figure}

\begin{lemma}
\label{lem:trapezoid}
    Our implementation of the trapezoid primitive requires $O(1)$ rounds.
\end{lemma}

\begin{remark}
    Using the shearing primitive, we can reflect the path between the bases starting from the starting point.
\end{remark}



\section{Monotone Structures}
\label{sec:monotone}

Both our universal reconfiguration algorithm and our constant-time Spiral2Line algorithm make use of monotone structures as intermediate structures.
In \Cref{sec:monotone2line}, we show how to reconfigure a monotone structure into a line segment.
In \Cref{sec:starconvex2monotone}, we adapt the algorithm to reconfigure histograms to a line segment within a specific area.
This will allow us to reconfigure a star-convex structure into a monotone structure.

\subsection{Monotone2LineSegment Algorithm}
\label{sec:monotone2line}


In this section, we present our Monotone2LineSegment algorithm.
W.l.o.g., we assume that the amoebot structure is $y$-monotone.
By definition, each column is a $y$-segment.
We first remove all bonds between adjacent columns except for one arbitrary horizontal bond.
If there is no such bond, we shift the right column up by one position in a preprocessing phase (see \Cref{fig:m2l:preprocessing}).

\begin{figure}[tb]
    \centering
    \includegraphics{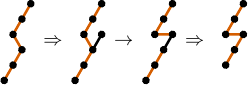}
    \caption{
        Preprocessing phase of the Monotone2LineSegment algorithm.
    }
    \label{fig:m2l:preprocessing}
    \label{fig:m2l:rule:0}
\end{figure}


Let $(v_1, \dots, v_m)$ denote the amoebots of the column from bottom to top.
Let $v_\ell$ ($v_r$) be the amoebot connected to the left (right) column if it exists and $v_\ell = v_1$ ($v_r = v_m$) otherwise.
Each column applies the following rules until $m = 1$ (see \Cref{fig:m2l:rule}).

\begin{figure}[tbp]    
    \begin{minipage}[t]{\linewidth}
        \centering
        \includegraphics{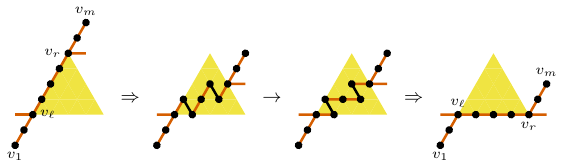}
        \subcaption{Rule 1 ($\ell < r$).}
        \label{fig:m2l:rule:1}
    \end{minipage}
    
    \begin{minipage}[t]{\linewidth}
        \centering
        \includegraphics{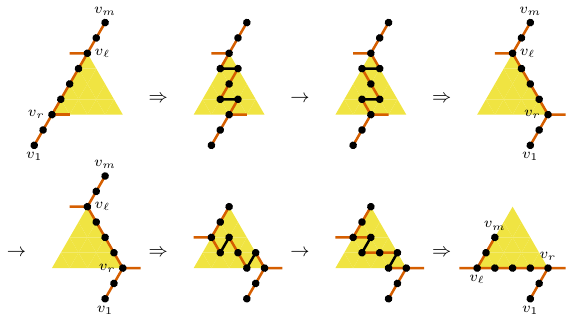}
        \subcaption{Rule 2 ($\ell > r$).}
        \label{fig:m2l:rule:2}
    \end{minipage}
    
    \begin{minipage}[t]{\linewidth}
        \centering
        \includegraphics{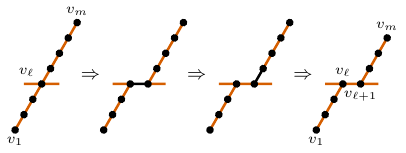}
        \subcaption{Rule 3a ($\ell = r$ and $1 < \ell < m$).}
        \label{fig:m2l:rule:3a}
    \end{minipage}
    
    \begin{minipage}[t]{\linewidth}
        \centering
        \includegraphics{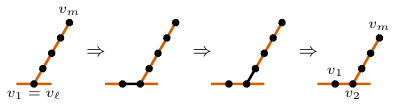}
        \subcaption{Rule 3b ($\ell = r$, $\ell = 1$ and $m$ odd).}
        \label{fig:m2l:rule:3b}
    \end{minipage}
    
    \begin{minipage}[t]{\linewidth}
        \centering
        \includegraphics{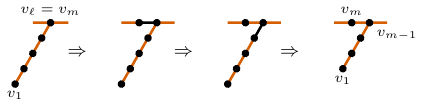}
        \subcaption{Rule 3c ($\ell = r$, $\ell = m$ and $m$ odd).}
        \label{fig:m2l:rule:3c}
    \end{minipage}
    
    \begin{minipage}[t]{\linewidth}
        \centering
        \includegraphics{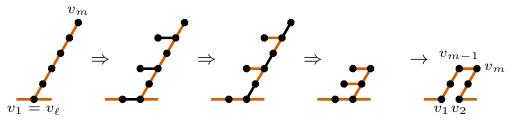}
        \subcaption{Rule 3d ($\ell = r$, $\ell = 1$ and $m$ even).}
        \label{fig:m2l:rule:3d}
    \end{minipage}
    
    \begin{minipage}[t]{\linewidth}
        \centering
        \includegraphics{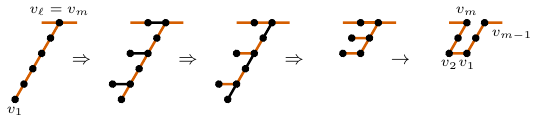}
        \subcaption{Rule 3e ($\ell = r$, $\ell = m$ and $m$ even).}
        \label{fig:m2l:rule:3e}
    \end{minipage}
    \caption{
        Rules used in the Monotone2LineSegment algorithm.
    }
    \label{fig:m2l:rule}
\end{figure}

\begin{enumerate}
    \item
    If $\ell < r$, we shear the subsegment $(v_\ell, \dots, v_r)$ to align it to the $x$-axis (see \Cref{fig:m2l:rule:1}).
    We obtain $r - \ell + 1$ new columns.
    The first column consists of $(v_1, \dots, v_\ell)$.
    For all $i \in \{ 2, \dots, r - \ell \}$, the $i$-th column consists of $v_{\ell - 1 + i}$.
    The last column consists of $(v_r, \dots, v_m)$.
    For each of these columns, $\ell = r$ holds.
    Note that if $\ell = 1 \land r = m$, all resulting columns will consist of a single amoebot, respectively.

    \item
    If $\ell > r$, we shear the subsegment $(v_r, \dots, v_\ell)$ twice to first align it to the $y$-axis and then to the $x$-axis (see \Cref{fig:m2l:rule:2}).
    We obtain $\ell - r + 1$ new columns.
    The first column consists of $(v_\ell, \dots, v_m)$.
    For all $i \in \{ 2, \dots, \ell - r \}$, the $i$-th column consists of $v_{\ell + 1 - i}$.
    The last column consists of $(v_1, \dots, v_r)$.
    For each of these columns, $\ell = r$ holds.
    Note that if $\ell = m \land r = 1$, all resulting columns will consist of a single amoebot, respectively.
    
    \item
    If $\ell = r$, we apply the following subrules.
    \begin{enumerate}
        \item 
        If $1 < \ell < m$, we split the column into two columns as follows (see \Cref{fig:m2l:rule:3a}).
        First, $v_\ell$ expands horizontally.
        Second, $v_\ell$ and $v_{\ell+1}$ perform a handover.
        Finally, $v_{\ell+1}$ contracts again.
        The left column consists of $(v_1, \dots, v_\ell)$.
        The right column consists of $(v_{\ell+1}, \dots, v_m)$.

        \item 
        If $\ell = 1$ and $m$ is odd, we split the column into two columns analogously to Rule 3a (see \Cref{fig:m2l:rule:3b}).
        The left column consists of $v_1$.
        The right column consists of $(v_2, \dots, v_m)$.

        \item 
        If $\ell = m$ and $m$ is odd, we split the column into two columns analogously to Rules 3a and 3b (see \Cref{fig:m2l:rule:3c}).
        The left column consists of $v_m$.
        The right column consists of $(v_2, \dots, v_{m-1})$.

        \item
        If $\ell = 1$ and $m$ is even, we split the subsegment $(v_\ell, \dots, v_m)$ into two columns with $m/2$ amoebots, respectively, as follows (see \Cref{fig:m2l:rule:3d}).
        For each $i \in \{ 1, 3, \dots, m-1\}$, $v_i$, we perform the following procedure in parallel.
        In the first round, each $v_i$ expands horizontally.
        In the second round, each $v_i$ performs a handover with $v_{i+1}$.
        In the third round, each $v_{i+1}$ contracts.
        We obtain two columns with $m/2$ amoebots.
        The left column consists of $(v_1, v_3, \dots, v_{m-1})$.
        The right column consists of $(v_2, v_4, \dots, v_m)$.
        
        For subsequent rule applications, we connect all amoebots within the same column with bonds and remove all bonds between those columns except for the topmost horizontal bond.
        Note that this does not require an additional round.

        \item
        If $\ell = m$ and $m$ is even, we split the subsegment $(v_\ell, \dots, v_m)$ into two columns with $m/2$ amoebots, respectively, analogously to Rules 3d (see \Cref{fig:m2l:rule:3e}).
        The left column consists of $(v_2, v_4, \dots, v_m)$.
        The right column consists of $(v_1, v_3, \dots, v_{m-1})$.
        
        For subsequent rule applications, we connect all amoebots within the same column with bonds and remove all bonds between those columns except for the bottommost horizontal bond.
        Note that this does not require an additional round.
    \end{enumerate}
\end{enumerate}

We obtain the following theorem.

\begin{theorem}
\label{th:monotone2line}
    The Monotone2LineSegment algorithm reconfigures a structure of {\sc Monotone} into a structure of {\sc LineSegment} in $16$ rounds.
\end{theorem}

\begin{proof}

    First, consider the preprocessing phase and each rule separately.
    We show that they can be performed without causing any collisions.
    In the preprocessing phase, no amoebot leaves its column.
    So, no collision can occur with amoebots of other columns.
    In each rule, the column pushes its adjacent columns to the left and right while creating new columns.
    This prevents collisions with other columns.
    Finally, note that there cannot be any collisions within a column since each column is connected at all times.
    Initially, this holds by definition since the amoebot structure is monotone.
    Afterwards, it holds by construction of the preprocessing phase and all rules.
    

    For the sake of analysis, we split Rules 1 and 2:
    Rule 1a includes all cases $\ell < r \land (\ell \neq 1 \lor r \neq m)$,
    Rule 1b all cases $\ell < r \land (\ell = 1 \land r = m)$,
    Rule 2a all cases $\ell > r \land (\ell \neq m \lor r \neq 1)$, and
    Rule 2b all cases $\ell > r \land (\ell = m \land r = 1)$.
    Note that we do not change how the rules are executed.


    Now, consider any amoebot.
    We construct a state graph for the amoebot as follows (see \Cref{fig:rule_graph}).
    The state of the amoebot is defined by the rule we need to apply on its column.
    This includes a state where no rule can be applied.
    Clearly, that state is also the termination state.
    We add an edge from Rule $x$ to Rule $y$ if and only if the amoebot can end up in a column, on which we need to apply Rule $y$, after we applied Rule $x$ on its current column.
    \Cref{fig:rule_graph} shows the resulting graph.
    Initially, an amoebot can be in any of the states including the termination state.
    However, note that the graph is acyclic.
    Hence, after at most $4$ rules, the amoebot will reach the termination state.
    
\begin{figure}[tb]
    \centering
    \includegraphics[scale=.5]{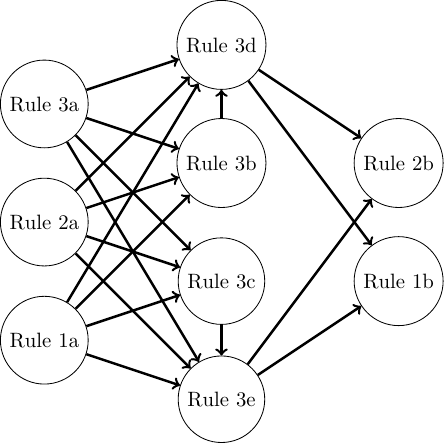}
    \caption{
        State graph.
        The termination state is not shown.
        Any other state has an edge to the termination state.
    }
    \label{fig:rule_graph}
\end{figure}

    
    Since the preprocessing step only takes $2$ rounds, each amoebot only participates in at most $4$ rules, each rule only requires $O(1)$ rounds, the algorithm terminates after $O(1)$ rounds.
    More precisely, the longest path in the state graph plus the preprocessing phase takes $16$ rounds.
    Since any amoebot in the termination state is the sole amoebot of its column and only has horizontal bonds, we obtain a line segment once all amoebots have terminated, i.e., the algorithm is correct.
\end{proof}



\subsection{StarConvex2Monotone Algorithm}
\label{sec:histogram2line}
\label{sec:starconvex2monotone}



In this section, we reconfigure star-convex structures into monotone structures.
For that, we first need to consider the reconfiguration of a histogram into a line segment within a specific area that we will define below.
Note that without this restriction, we can apply the Monotone2LineSegment algorithm to reconfigure the histogram into a line segment since histograms are monotone.

W.l.o.g., consider a histogram that consists of an $x$-segment and attached $y$-segments on the top.
Our goal is to reconfigure it into line segment such that no amoebot moves beyond the $x$- and $y$-axis through the leftmost amoebot on the $x$-segment.
During the reconfiguration of star-convex structures into monotone structures, we will use the leftmost amoebot as a connection point.

We adapt the Monotone2LineSegment algorithm as follows.
First, we choose the bonds along the $x$-segment as the horizontal bonds between the columns.
Second, we combine Rules 1, 2, and 3d to a new Rule:
\begin{enumerate}
    \setcounter{enumi}{3}
    \item 
    If $\ell = 1$ and $m$ is even, we split the column into $m$ columns with a single amoebot, respectively, as follows (see \Cref{fig:m2l:rule:4}).
    The rule consists of $3$ phases.
    In the first phase, we apply Rule 3d.
    In the second phase, we apply the first shear primitive of Rule 2 on the right column.
    In the third phase, we apply Rule 1 on the left column and the second shear of Rule 2 on the right column in parallel.
\end{enumerate}
We obtain the following lemma.

\begin{figure}[tb]
    \centering
    \includegraphics{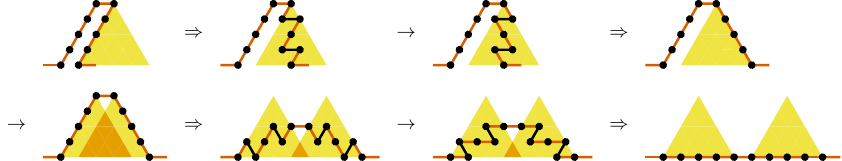}
    \caption{
        Rule 4.
        We omit the first phase, i.e., the application of Rule 3d (see \Cref{fig:m2l:rule:3d}), and only show the second and third phase.
    }
    \label{fig:m2l:rule:4}
\end{figure}

\begin{lemma}
\label{lem:histogram2line}
    The Histogram2LineSegment algorithm, i.e., the adapted Monotone2Line algorithm, reconfigures a structure of {\sc Histogram} into a structure of {\sc LineSegment} in $O(1)$ rounds without moving any amoebot beyond the $x$- and $y$-axis through the leftmost amoebot.
\end{lemma}

\begin{proof}
    The correctness and runtime follows mostly from \Cref{th:monotone2line}.
    Note that Rule 4 cannot cause any collisions since it is a combination of Rules 1, 2, and 3d.
    It only remains to prove that we do not move any amoebot beyond the $x$- and $y$-axis through the leftmost amoebot.
    For that, we argue that each rule only pushes amoebots to the right (and with that away from the $y$-axis through the leftmost amoebot) without moving any amoebots below the $x$-axis through the leftmost amoebot.
    Note that we only need to consider Rules 3b and 4 since initially, all horizontal bonds are the bottommost ones of each column, and both rules maintain this property.
    By construction, both rules only push amoebots to the right (and with that away from the $y$-axis through the leftmost amoebot) without moving any amoebots below the $x$-axis through the leftmost amoebot.
\end{proof}



Now, we can describe our StarConvex2Monotone algorithm that reconfigures a star-convex structure, w.l.o.g., into a $y$-monotone structure.
By definition, we can split a star-convex structure into the center amoebot and at most $6$ disjoint histograms (see \Cref{fig:star_convex}).
We apply the Histogram2LineSegment algorithm on each of these histograms with the amoebots adjacent to the center amoebot as connection points such that we obtain $2$ $y$-segments and $4$ $x$-segments.
We obtain the following theorem.

\begin{figure}[tb]
    \centering
    \includegraphics{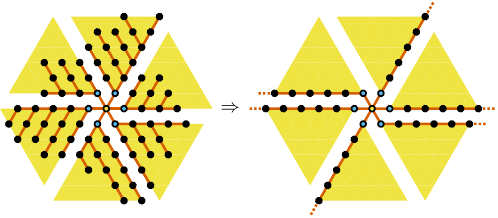}
    \caption{
        Histograms in a star-convex structures.
        The yellow amoebot indicates the center amoebot.
        The blue amoebots indicate the connection points used by the Histogram2LineSegment algorithm.
        The yellow triangles indicate the area for each instance from which no amoebot is moved out, respectively.
        The double arrow indicates multiple rounds.
    }
    \label{fig:star_convex}
\end{figure}

\begin{theorem}
    The StarConvex2Monotone algorithm reconfigures a structure of {\sc StarConvex} into a structure of {\sc Monotone} in $O(1)$ rounds.
\end{theorem}

\begin{proof}
    Note that by \Cref{lem:histogram2line}, the $6$ instances of the Histogram2LineSegment algorithm cannot collide with each other since they stay within their part of the grid (see \Cref{fig:star_convex}).
    The runtime directly follows from \Cref{lem:histogram2line}.

    It remains to prove the correctness, i.e., that we obtain a $y$-monotone structure.
    The $2$ resulting $y$-segments form a single $y$-segment with the center amoebot.
    This segment is trivially $y$-monotone.
    The $4$ resulting $x$-segments form two pairs of adjacent $x$-segments attached to $y$-segment.
    Since the $x$-segments are adjacent, they must be $y$-monotone.
    Since all three parts are $y$-monotone, the structure is $y$-monotone.
\end{proof}


\section{Universal Reconfiguration}
\label{sec:universal}

In this section, we describe our universal reconfiguration algorithm that reconfigures an arbitrary structure into a line segment.
It consists of $3$ subroutines.
In the first subroutine, we apply anthe Arbitrary2Bounded algorithm (called Arbitrary2Bounded) to reconfigure the structure into a bounded structure.
In the second subroutine, we apply an algorithm (calledthe Bounded2Monotone) algorithm to reconfigure the bounded structure into a monotone structure.
Finally, in the third subroutine, we apply the Monotone2LineSegment algorithm to reconfigure the monotone structure into a line segment.
The latter subroutine was presentedshown in \Cref{sec:monotone2line}.
In the following, we present the other two subroutines.


First, consider the \textbf{reconfiguration into a bounded structure}.
W.l.o.g., we will bound the number of rows to $\sqrt n$.
The idea is to iteratively merge rows with less than $\sqrt n$ amoebots with neighboring rows.
Each iteration proceeds as follows.
We first pair a maximum number of adjacent rows such that each pair has at least one row with less than $\sqrt n$ amoebots.


Then, we merge each pair of rows as follows.
The merge consists of $3$ phases.
In the first phase, each amoebot expands into the other row parallel to the $y$-axis if the node is unoccupied (see \Cref{fig:universal:expand}).
Observe that the amoebots that are not able to expand form pairs (see the blue amoebots in \Cref{fig:universal:expand}).


In the second phase, we iteratively move one amoebot of each of these pairs to the next unoccupied position in the row.
W.l.o.g., consider the rightmost pair of amoebots and all expanded amoebots to the next unoccupied position to the right (see \Cref{fig:universal:tunneling}).
The idea is to construct an alternating chain from the pair of amoebots to the next unoccupied position.
For that, we contract every second amoebot to its bottom node, ensuring the rightmost amoebot remains expanded.
This construction guarantees that we can maintain connectivity within the amoebot structure since the participating amoebots can keep at least one bond to each of the initially adjacent nodes (see the gray nodes in \Cref{fig:universal:tunneling}).
Now, we can apply the tunneling primitive to move all amoebots to the right (see \Cref{lem:tunneling}).
Finally, we expand all contracted amoebots again.
We continue with the next iteration if there is another pair of amoebots.
Otherwise, we continue to the next phase.


In the third phase, we contract all expanded amoebots again (see \Cref{fig:universal:contract}).
This concludes the merge.
We obtain the following lemma.


\begin{figure}[tb]
    \centering
    \includegraphics{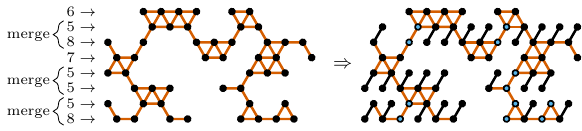}
    \caption{
        First phase of the Arbitrary2Bounded algorithm.
        The amoebot structure has $n = 49$.
        Hence, each pair of rows that we want to merge has at least one row with less than $\sqrt n = 7$ amoebots.
        The blue amoebots indicate the pairs of amoebots that are unable to expand.
    }
    \label{fig:universal:expand}
\end{figure}

\begin{figure}[tb]
    \begin{minipage}[t]{\linewidth}
        \centering
        \includegraphics{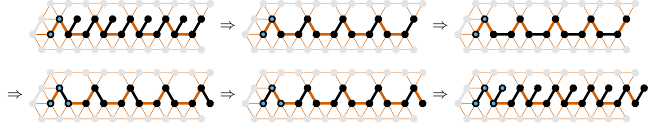}
        \subcaption{Even number of amoebots.}
        \label{fig:universal:tunneling:even}
    \end{minipage}
    
    \begin{minipage}[t]{\linewidth}
        \centering
        \includegraphics{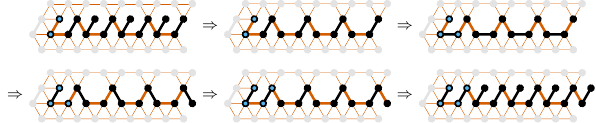}
        \subcaption{Odd number of amoebots.}
        \label{fig:universal:tunneling:odd}
    \end{minipage}
    \caption{
        Single iteration of the second phase of the Arbitrary2Bounded algorithm.
        The blue amoebots indicate a pair of amoebots that were unable to expand in the first phase.
        The gray nodes are initially adjacent to the participating amoebots.
    }
    \label{fig:universal:tunneling}
\end{figure}

\begin{figure}[tb]
    \centering
    \includegraphics{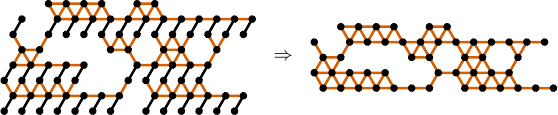}
    \caption{
        Third phase of the Arbitrary2Bounded algorithm.
        All amoebots contract again.
    }
    \label{fig:universal:contract}
\end{figure}


\begin{lemma}
\label{lem:arbitrary2bounded}
    The Arbitrary2Bounded algorithm reconfigures a structure of {\sc Arbitrary} into a structure of {\sc Bounded($\sqrt n$)} in $O(\sqrt n \log n)$ rounds.
\end{lemma}

\begin{proof}
    We first show that the algorithm does not cause any collisions.
    In the first phase, amoebots only expand into unoccupied nodes.
    In the second phase, amoebots only perform isolated handovers.
    In the third phase, we contract whole rows.
    Hence, neither phase can cause a collision.
    Since the algorithm does not cause any collisions, the correctness follows from the fact that there can be at most $\sqrt n$ rows with at least $\sqrt n$ amoebots.


    Since we can pair all rows with less than $\sqrt n$ amoebots (except for at most $1$), the number of rows with less than $\sqrt n$ amoebots gets reduced by (approximately) $\frac{1}{2}$.
    Therefore, after $O(\log n)$ iterations, all rows have at least $\sqrt n$ amoebots such that we terminate.
    
    
    Each iteration consists of three phases.
    The first and third phase only consist of a single round, respectively.
    A single iteration of the second phase requires $5$ rounds: $1$ round to construct the alternating chain, $3$ rounds to apply the tunneling primitive (see \Cref{lem:tunneling}), and $1$ round to expand the amoebots again.
    
    
    It remains to bound the number of necessary iterations in the second phase.
    Since at least one row of each pair of rows has less than $\sqrt n$ amoebots, there can only be at most $\sqrt n - 1$ many pairs of amoebots that cannot expand in the first phase.
    Therefore, we only need at most $\sqrt n - 1$ iterations in the second phase.
    %
    %
    Overall, we require $O(\sqrt n \log n)$ rounds.
\end{proof}

\begin{remark}
    The second phase can be parallelized within each row.
    However, if all pairs of amoebots are next to each other, we can only tunnel the two outermost pairs at the same time.
    In the worst case, the parallelization can only speed up the second phase of the algorithm by a factor of $2$.
    Therefore, the runtime does not improve asymptotically.
\end{remark}




Second, consider the \textbf{reconfiguration into a monotone structure}.
W.l.o.g., we assume that the number of rows in the amoebot structure is bounded by $k$ and our goal is to reconfigure it into a $y$-monotone structure.
We make use of the combing algorithm by Aloupis \emph{et al.} \cite{DBLP:journals/comgeo/AloupisCDDFLORAW09}.
The algorithm utilizes a sweeping approach.
The sweep line\footnote{In \cite{DBLP:journals/comgeo/AloupisCDDFLORAW09}, the sweep line is called a \emph{wall}.} moves from the topmost row to the bottommost row.
We remove all bonds not parallel to the $y$-axis above of the sweep line such that we obtain $y$-segments that are only connected to the sweep line.
Whenever the sweep line moves from one row to the next row, we move each of these line segments downwards if the node beneath the line segment is unoccupied.
For that, we simply expand the bottommost amoebot of the line segment to that node and contract it again afterwards (see \Cref{fig:universal:sweep}).

\begin{figure}[tb]
    \centering
    \includegraphics{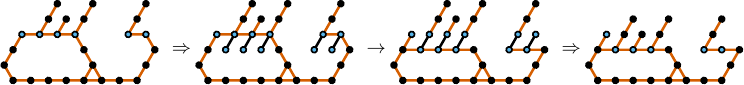}
    \caption{
        Combing algorithm by Aloupis \emph{et al.} \cite{DBLP:journals/comgeo/AloupisCDDFLORAW09}.
        The blue amoebots indicate the sweep line. 
    }
    \label{fig:universal:sweep}
\end{figure}

\begin{lemma}
\label{lem:bounded2monotone}
    The Bounded2Histogram algorithm reconfigures a structure of {\sc Bounded($k$)} into a structure of {\sc Monotone} in $O(k)$ rounds.
\end{lemma}

\begin{proof}
    We can simulate any (reconfiguration) algorithm for other models of modular robot systems if we can perform the same movement primitives.
    Hence, the correctness and runtime follows from \cite{DBLP:journals/comgeo/AloupisCDDFLORAW09}.
\end{proof}


Finally, the universal reconfiguration algorithm applies the Monotone2Line algorithm to reconfigure the monotone structure to a line segment (see \Cref{sec:monotone2line}).
By combining \Cref{lem:arbitrary2bounded,lem:bounded2monotone,th:monotone2line}, we obtain the following theorem.

\begin{theorem}
\label{th:arbitrary2line}
    The universal reconfiguration algorithm reconfigures a structure of {\sc Arbitrary} into a structure of {\sc LineSegment} in $O(\sqrt n \log n)$ rounds.
\end{theorem}


\section{Spirals}
\label{sec:spirals}

In this section, we present two algorithms to reconfigure a spiral into a line.
In \Cref{sec:spiral:unrolling}, we describe a simple algorithm that preserves the order of the amoebots within the spiral.
In \Cref{sec:spiral:constant}, we show a faster algorithm that only requires constant time.


\subsection{Unrolling}
\label{sec:spiral:unrolling}

In this section, we present our unrolling algorithm to reconfigure a spiral into a line while preserving the order of the amoebots (see \Cref{fig:unrolling}).
The algorithm iteratively removes corners by aligning the outermost segment to the next segment.
For that, we simply apply the shearing primitive.
We terminate once all corners were removed.
We obtain the following theorem.

\begin{figure}[tb]
    \centering
    \includegraphics{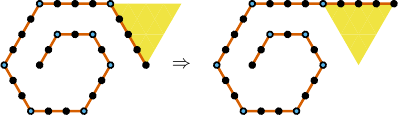}
    \caption{
        Unrolling algorithm.
        The blue amoebots indicate the corners in the initial amoebot structure.
        The double arrow indicates a phase consisting of $2$ rounds.
    }
    \label{fig:unrolling}
\end{figure}

\begin{theorem}
\label{th:unrolling}
    The unrolling algorithm reconfigures a structure of {\sc Spiral} into a structure of {\sc LineSegment} while preserving the order of the amoebots in $O(c)$ rounds where $c$ denotes the number of corners.
\end{theorem}

\begin{proof}
    Note that the shearing primitive preserves the order of the amoebots.
    Therefore, correctness follows from induction since each iteration removes one corner while preserving the order.
    Each iteration only consists of a single application of the shearing primitive.
    By \Cref{lem:shearing}, this requires $2$ rounds.
    The number of iteration is bounded by the number, $c$, of corners.
    Overall, the algorithm requires $O(c)$ rounds.
\end{proof}

\begin{remark}
    We cannot accelerate the algorithm by applying the shearing primitive on multiple segments in parallel because the number of shearing primitives the initially outermost segment has to participate in is linear in the number of corners.
    Furthermore, parallelization can lead to collisions.
\end{remark}



\subsection{Constant-Time Algorithm}
\label{sec:spiral:constant}


In this section, we present our Spiral2LineSegment algorithm.
We assume that the spiral consists of at least $5$ line segments.
Otherwise, we resort to the unrolling algorithm which gives us a constant-time reconfiguration in this case.
W.l.o.g., we assume that the spiral starts with a $y$-segment, spirals outwards in a clockwise direction, and ends with an $x$-segment.
We may perform up to $2$ iterations of the unrolling algorithm to enforce the latter assumption.


Similar to the universal reconfiguration algorithm, we go through a monotone structure.
Hence, the algorithm starts with the Spiral2Monotone algorithm which reconfigures the spiral into a $y$-monotone structure.
The idea is to form a base line through the spiral which allows us to cut the ``arcs'' of the spiral without disconnecting the structure.
In order to then obtain a monotone structure, we only need to align all segments that are not part of the base line, w.l.o.g., with the $y$-axis.
In the following, we explain both parts in more detail.


We first define a \emph{base line} as follows (see \Cref{fig:spiral:base_line}).
We start with the first three innermost segments.
From each left (right) corner, we go straight to the left (right) until we hit the next line segment if such line segment exists.
From there, we go to the next left (right) corner of that line segment.
If such a line segment does not exists, we stop the base line at the corner.
Note that the base line is $y$- and $z$-monotone by construction.

\begin{figure}[tb]
    \centering
    \includegraphics{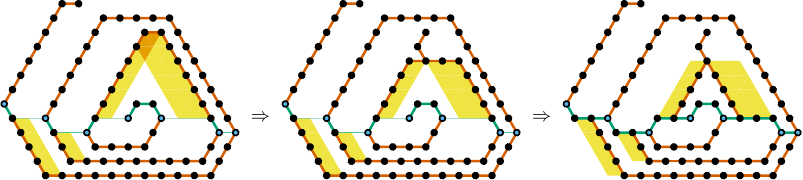}
    \caption{
        Base line construction.
        The blue amoebots indicate the left and right corners.
        The green path indicates the base line.
        The yellow areas indicate the parallelograms.
        The orange areas indicate the areas where parallelograms intersect.
        Each double arrow indicates multiple rounds.
    }
    \label{fig:spiral:base_line}
\end{figure}

In order to construct the base line, we need to fill the paths from the left (right) corners to the next line segments to the left (right).
Observe that the remaining base line is already occupied by amoebots.
The idea is to apply the parallelogram primitive on each hit line segment and its adjacent $x$-segment in parallel.
However, the parallelograms may intersect with each other or may degenerate because the $x$-segments are too short.
To resolve these issues, we can apply the trapezoid primitive.
We choose the starting point such that it becomes an acute corner of the parallelogram.


The base line allows us to disconnect each $x$-segment (except for the innermost one, which is part of the base line) at one of its endpoints without disconnecting the amoebot structure.
Note that some $x$-segments could been reduced to a single amoebot.
Similar to the triangle primitive (see \Cref{lem:triangle}), we will call all line segments that are not part of the base line ``arms''.
It remains to align the arms to the $y$-axis with the shearing primitive.
However, before we do that, we must ensure that we have enough space to perform all applications of the primitive.


We start by aligning the $z$-segments in the arms in the following $3$ phases (see \Cref{fig:spiral:z}).
In the first phase, we expand all amoebots on the base line connected to a $z$-segment of an arm horizontally.
In the second phase, we can apply the shearing primitive on each $z$-segment.
In the third phase, we contract all amoebots again.

\begin{figure}[tb]
    \centering
    \includegraphics{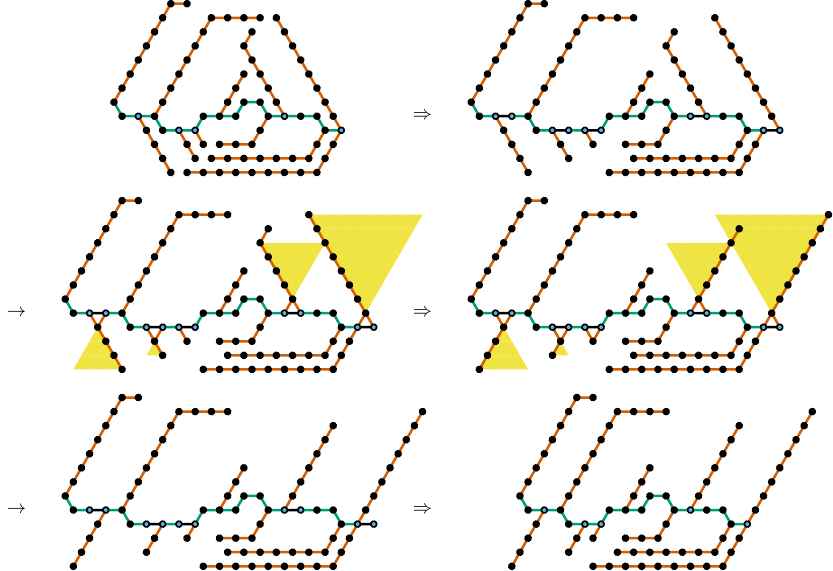}
    \caption{
        Alignment of the $z$-segments.
        The blue amoebots indicate the amoebots on the base line connected to a $z$-segment of an arm.
        The green path indicates the base line.
        The yellow areas indicate the shearing primitive.
        The second double arrow indicates $2$ rounds.
    }
    \label{fig:spiral:z}
\end{figure}


Next, we align the $x$-segments in the arms in the following $3$ phases (see \Cref{fig:spiral:x}).
In the first phase, we expand all left and right corners on the base line parallel to the $z$-axis, and all corners in the arms parallel to the $z$-axis.
In the second phase, we can apply the shearing primitive on each $x$-segment.
In the third phase, we contract all amoebots again.
We obtain the following lemma.

\begin{figure}[t!]
    \centering
    \includegraphics{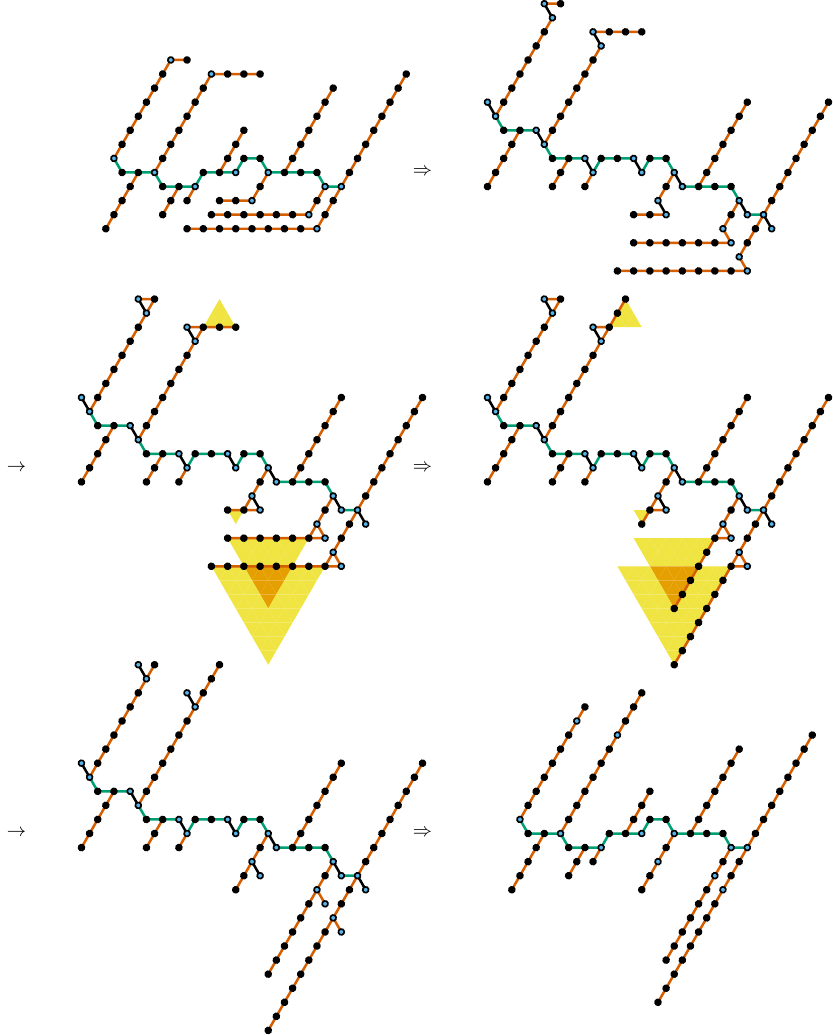}
    \caption{
        Alignment of the $x$-segments.
        The blue amoebots indicate all left and right corners on the base line and all corners in the arms.
        The green path indicates the base line.
        The yellow areas indicate the shearing primitive.
        The second double arrow indicates $2$ rounds.
    }
    \label{fig:spiral:x}
\end{figure}


\begin{lemma}
\label{lem:spiral2monotone}
    The Spiral2Monotone algorithm reconfigures a structure of {\sc Spiral} into a structure of {\sc Monotone} in $O(1)$ rounds.
\end{lemma}

\begin{proof}

    Since the base line is $y$-monotone by definition and we align all arms to the $y$-axis, correctness follows by construction if we can prove that no collision can occur.
    During the formation of the base line, we only apply the trapezoid and parallelogram primitive, which function in-place, i.e., no collision can occur.

    We can divide the base line into three parts:
    the center part that consists of the three innermost segments, the left part that consists of all amoebots to the left of the center part, and the right part that consists of all amoebots to the right of the center part.
    Furthermore, we can divide the arms into four groups:
    \begin{enumerate}
        \item Each arm attached to the top of the left part of the base line consists of a $y$-segment and possibly an $x$-segment.
        \item Each arm attached to the bottom of the left part of the base line consists of a $z$-segment and possibly a $y$-segment (resulting from the trapezoid primitive).
        \item Each arm attached to the bottom of the right part of the base line consists of a $y$-segment and possibly an $x$-segment.
        \item Each arm attached to the top of the right part of the base line consists of a $z$-segment and possibly a $y$-segment (resulting from the trapezoid primitive).
    \end{enumerate}
    Consider the connection points of the connection points of the arms with base line.
    Note that each connection point is either one of the original left or right corners, or the obtuse corner of one of the (possibly degenerate) parallelograms.
    This implies that following properties:
    \begin{enumerate}
        \item For each arm attached to the bottom of the left part of the base line, the next base line amoebot to the left of its connection point cannot be below the connection point.
        \item For each arm attached to the top of the right part of the base line, the next base line amoebot to the right of its connection point cannot be above the connection point.
    \end{enumerate}
    For the remaining algorithm, i.e., the alignment of the $x$- and $z$-segments, we need to show that neither the base line nor any group collides internally, and that no group collides with another group or the base line.
    In general, two groups cannot collide with each other if they are on different sides of the base line without colliding with the base line itself.
    Hence, it suffices to only consider whether a group can collide with the base line and the neighboring group on the same side of the base line.


    Consider the first phase of the alignment of the $z$-segments.
    The expansions move the arms apart.
    The base line cannot collide internally since it is monotone.
    Since we move the arms apart from each other and the form of the arms (two line segments with a $120^\circ$ degree angle), the groups cannot collide internally, with each other, or the base line.

    
    Consider the second phase of the alignment of the $z$-segments.
    We only have to argue that the second and fourth group of arms do not collide internally, with the base line, or with other groups since the base line and the other group do not move.
    In this phase, we apply the shearing primitive on the $z$-segments.
    Since these segments are initially parallel and shear with the same speed, they cannot collide internally.
    They also cannot collide with the base line since the base line is monotone and therefore, cannot be in the shearing area.
    Finally, they cannot collide with the neighboring group since they shear away from it.

    
    Consider the third phase of the alignment of the $z$-segments.
    The contractions move the arms that we moved apart together again.
    The base line cannot collide internally since it is monotone.
    Since the arms have either not changed or they have changed in the same way, the groups cannot collide internally or with each other.
    Furthermore, the second and fourth group cannot collide with the base line due to the aforementioned properties about their connection points.
    

    Consider the first phase of the alignment of the $x$-segments.
    For now, only consider the expansions on the base line.
    These expansions move the arms apart.
    The base line cannot collide internally since it is monotone.
    Since we move the arms apart from each other and the form of the arms (two line segments with a $120^\circ$ degree angle), the groups cannot collide internally, with each other, or the base line.
    Now, consider the additional expansions in the arm.
    Since we moved the arms apart, we have enough space to shift the $x$-segments in the first group upwards and the $x$-segments in the third group downwards.

    
    Consider the second phase of the alignment of the $x$-segments.
    We only have to argue that the first and third group of arms do not collide internally, with the base line, or with other groups since the base line and the other group do not move.
    In this phase, we apply the shearing primitive on the $x$-segments.
    Since these segments are initially parallel and shear with the same speed, they cannot collide internally.
    Finally, they cannot collide with the neighboring group or the base line since they shear away from the neighboring group or the base.

    
    Consider the third phase of the alignment of the $x$-segments.
    The contractions move the arms that we moved apart together again.
    The base line cannot collide internally since it is monotone.
    Since the arms have either not changed or they have changed in the same way, the groups cannot collide internally or with each other.
    Furthermore, the first and third group cannot collide with the base line since the shearing did not change the connectivity at the connection points.
    Altogether, no collision can occur, i.e., the algorithm is correct.


    It remains to analyze the runtime.
    By \Cref{lem:parallelogram,lem:trapezoid}, the formation of the base line requires $O(1)$ rounds.
    By \Cref{lem:shearing}, the alignment of the $x$- and $z$-segments requires $4$ rounds (including the set up and the restoration of the base line), respectively.
    Overall, the algorithm requires $O(1)$ rounds.
\end{proof}

By combining \Cref{lem:spiral2monotone,th:monotone2line}, we obtain the following theorem.

\begin{theorem}
\label{th:spiral2line}
    The Spiral2LineSegment algorithm reconfigures a structure of {\sc Spiral} into a structure of {\sc LineSegment} in $O(1)$ rounds.
\end{theorem}





\section{Open Problems}
\label{sec:conclusion}


In this paper, we have shown how to reconfigure any structure into another in sublinear time.
It is still open whether a polylogarithmic time universal reconfiguration algorithm is possible.
One way to achieve such algorithm could be an algorithm that reconfigures a structure of {\sc Arbitrary} into a structure of {\sc Bounded($\operatorname{polylog} n$)}, which would reduce the runtime of our universal reconfiguration algorithm to polylogarithmic time.
Another way would be an algorithm that reconfigures a structure of {\sc Arbitrary} into a structure of {\sc Rhombical} for which we already know how to reconfigure it to any other structure of {\sc Rhombical} in logarithmic time \cite{DBLP:conf/isaac/AloupisCDLAW08,DBLP:journals/arobots/PadalkinKS25}.
%
%
We have also presented an algorithm that reconfigures a spiral into a line segment in constant time.
It would be interesting to investigate which other classes of structures can be reconfigured into a line segment in constant time.
A next step could be to consider arbitrary paths and trees.


Furthermore, we left the design of distributed algorithms for future work.
Recall that local communication between amoebots leads to a natural lower bound of $\Omega(D)$, where $D$ is the diameter of the structure.
Therefore, other ways of communication are required, e.g., the \emph{reconfigurable circuit extension} \cite{DBLP:journals/jcb/FeldmannPSD22} which enables fast global communication between amoebots.
Previous work has shown that reconfigurable circuits enable polylogarithmic-time solutions to various stationary problems \cite{DBLP:conf/wdag/ArtmannPS25,DBLP:journals/jcb/FeldmannPSD22,DBLP:conf/podc/PadalkinS24,DBLP:journals/nc/PadalkinSW24}.
Almalki \emph{et al.} \cite{DBLP:conf/sss/AlmalkiGMP25} were the first to leverage them for rapid, distributed transformations in the growth model.
It is open whether they can also prove of use in the amoebot model with joint movements.


\bibliography{reference}



\end{document}